\documentclass[onecolumn,12pt,draftclsnofoot,peerreview,oneside]{IEEEtran}
\usepackage{definitions}

\begin{document}
\title{Order Statistics Based List Decoding Techniques for Linear Binary Block
  Codes} \author{Saif E. A. Alnawayseh, and
  Pavel~Loskot,~\IEEEmembership{Member,~IEEE}%
  \thanks{The authors are with College of Engineering, Swansea University,
    Singleton Park, Swansea SA2 8PP, United Kingdom (email:
    \{491959,p.loskot\}@swan.ac.uk)}\thanks{Corresponding author: Pavel Loskot,
    phone: +44 1792 602619, fax: +44 1792 295676}\thanks{This work was
    presented in part at the \emph{IEEE Wireless Communications \& Signal
      Processing Conference (WCSP)}, Nanjing, China, Nov. 2009.}}
\markboth{IEEE Transactions on Information Theory}{Alnawayseh and Loskot: Order
  Statistics Based List Decoding}
\maketitle

\begin{abstract}
  The order statistics based list decoding techniques for linear binary block
  codes of small to medium block length are investigated. The construction of
  the list of the test error patterns is considered. The original order
  statistics decoding is generalized by assuming segmentation of the most
  reliable independent positions of the received bits. The segmentation is
  shown to overcome several drawbacks of the original order statistics
  decoding. The complexity of the order statistics based decoding is further
  reduced by assuming a partial ordering of the received bits in order to avoid
  the complex Gauss elimination. The probability of the test error patterns in
  the decoding list is derived. The bit error rate performance and the decoding
  complexity trade-off of the proposed decoding algorithms is studied by
  computer simulations. Numerical examples show that, in some cases, the
  proposed decoding schemes are superior to the original order statistics
  decoding in terms of both the bit error rate performance as well as the
  decoding complexity.
\end{abstract}

\newpage

\section{Introduction}

A major difficulty in employing forward error correction (FEC) coding is the
implementation complexity especially of the decoding at the receiver and the
associated decoding latency for long codewords. Correspondingly, the FEC coding
is often designed to trade-off the bit error rate (BER) with the decoding
complexity and latency. Many universal decoding algorithms have been proposed
for the decoding linear binary block codes \cite{Yagi05}. The decoding
algorithms in \cite{Fos95}--\cite{Val04} are based on the testing and
re-encoding of the information bits as initially considered by Dorsch in
\cite{Dorsch}. In particular, a list of the likely transmitted codewords is
generated using the reliabilities of the received bits, and then, the most
likely codeword is selected from this list. The list of the likely transmitted
codewords can be constructed from a set of the test error patterns. The test
error patterns can be predefined as in \cite{Fos95} and \cite{Gaz97} and in
this paper, predefined and optimized for the channel statistics as in
\cite{Kab07}, or defined adaptively for a particular received sequence as
suggested in \cite{Yagi06}. The complexity of the list decoding can be further
reduced by the skipping and stopping rules as shown, for example, in
\cite{Fos95} and \cite{Gaz97}.

Among numerous variants of the list decoding techniques, the order statistics
decoding (OSD) is well-known \cite{Fos95}, \cite{Gaz97}. The structural
properties of the FEC code are utilized to reduce the OSD complexity in
\cite{Fos96}. The achievable coding gain of the OSD is improved by considering
the multiple information sets in \cite{Fos02-1}. The decoding proposed in
\cite{Jin07-1} exploits randomly generated biases to present the decoder with
the multiple received soft-decision values. The sort and match decoding of
\cite{Val99} divides the received sequence into two disjoint segments. The list
decoding is then performed for each of the two segments independently, and the
two lists are combined using the sort and match algorithm to decide on the most
likely transmitted codeword. The box and match decoding strategy is developed
in \cite{Val04}. An alternative approach to the soft-decision decoding of
linear binary block codes relies on the sphere decoding techniques
\cite{Most09,Vikalo06}. For example, the input sphere decoder (ISD) discussed
in this paper can be considered to be a trivial sphere decoding algorithm.

In this paper, we investigate the OSD-based decoding strategies for linear
binary block codes. Our aim is to obtain low-complexity decoding schemes that
provide sufficiently large or valuable coding gains, and most importantly, that
are well-suited for implementation in communication systems with limited
hardware resources, e.g., at nodes of the wireless sensor network. We modify
the original OSD by considering the disjoint segments of the most reliable
independent positions (MRIPs). The segmentation of the MRIPs creates
flexibility that can be exploited to fine tune a trade-off between the BER
performance and the decoding complexity. Thus, the original OSD can be
considered to be a special case of the segmentation-based OSD having only one
segment corresponding to the MRIPs. Since the complexity of obtaining a row
echelon form of the generator matrix for every received codeword represents a
significant part of the overall decoding complexity, we examine a partial-order
statistics decoding (POSD) when only the systematic part of the received
codeword is ordered.

This paper is organized as follows. System model is described in Section
II. Construction of the list of test error patterns is investigated in Section
III. The list decoding algorithms are developed in Section IV. The performance
analysis is considered in Section V. Numerical examples to compare the BER
performance and the decoding complexity of the proposed decoding schemes are
presented in Section VI. Finally, conclusions are given in Section VII.

\section{System Model}

Consider transmission of codewords of a linear binary block code $\C$ over an
additive white Gaussian noise (AWGN) channel with Rayleigh fading. The code
$\C$, denoted as $(N,K,\dmin)$, has block length $N$, dimension $K$, and the
minimum Hamming distance between any two codewords $\dmin$. Binary codewords
$\cv\in\Z_2^N$ where $\Z_2=\{0,1\}$ are generated from a vector of information
bits $\uv\in\Z_2^K$ using the generator matrix $\Gv\in\Z_2^{K\times N}$, i.e.,
$\cv= \uv\Gv$, and all binary operations are considered over a Galois field
$\mathrm{GF}(2)$. If the code $\C$ is systematic, the generator matrix has the
form, $\Gv=[\I\ \Pv]$, where $\I$ is the $K\times K$ identity matrix, and
$\Pv\in\Z_2^{K\times(N-K)}$ is the matrix of parity checks. The codeword $\cv$
is mapped to a binary phase shift keying (BPSK) sequence $\xv\in\{+1,-1\}^N$
before transmission using a mapping, $x_i=\MP{c_i}=(-1)^{c_i}$, for
$i=1,2,\cdots,N$. Assuming bits $u_i$ and $u_j$, the mapping $\Mp$ has the
property,
\begin{equation}\label{eq:5}
  \MP{u_i\oplus u_j}= \MP{u_i}\MP{u_j}
\end{equation}
where $\oplus$ denotes the modulo $2$ addition. The encoded bit $c_i$ can be
recovered from the symbol $x_i$ using the inverse mapping, $c_i=
\MPi{x_i}=(1-x_i)/2$. For brevity, we also use the notation, $\xv=\MP{\cv}$ and
$\cv=\MPi{\xv}$, to denote the component-wise modulation mapping and
de-mapping, respectively. The code $\C$ is assumed to have equally probable
values of the encoded bits, i.e., the probability, $\Prob{c_i=0}=\Prob{c_i=1}
=1/2$, for $i=1,2,\cdots,N$. Consequently, all the codewords are transmitted
with the equal probability, i.e., $\Prob{\cv}=2^{-K}$ for $\forall \cv\in\C$.

The signal at the output of the matched filter at the receiver can be written
as,
\begin{equation*}
  y_i= h_i x_i+w_i
\end{equation*}
where the frequency non-selective channel fading coefficients $h_i$ as well as
the AWGN samples $w_i$ are mutually uncorrelated zero-mean circularly symmetric
complex Gaussian random variables. The variance of $h_i$ is unity, i.e.,
$\E{|h_i|^2}=1$ where $\E{\cdot}$ is expectation, and $|\cdot|$ denotes the
absolute value. The samples $w_i$ have the variance, $\E{|w_i|^2}= (R
\gammab)^{-1}$, where $R=K/N$ is the coding rate of $\C$, and $\gammab$ is the
signal-to-noise ratio (SNR) per transmitted encoded binary symbol. The
covariance, $\E{h_i h_j^\ast}=0$ for $i\neq j$, where $(\cdot)^\ast$ denotes
the complex conjugate, corresponds to the case of a fast fading channel with
ideal interleaving and deinterleaving. For a slowly block-fading channel, the
covariance, $\E{h_i h_j^\ast}=1$ for $\forall i,j=1,2,\cdots,N$, and the fading
coefficients are uncorrelated between transmissions of adjacent codewords.

In general, denote as $f(\cdot)$ the probability density function (PDF) of a
random variable. The reliability $r_i$ of the received signal $y_i$ corresponds
to a ratio of the conditional PDFs of $y_i$ \cite{Ben99}, i.e.,
\begin{equation*}
  \frac{f(y_i|x_i=+1,h_i)}{f(y_i|x_i=-1,h_i)}\propto \real{h_i^\ast y_i}=r_i
\end{equation*}
since the PDF $f(y_i|x_i,h_i)$ is conditionally Gaussian. Thus, the reliability
$r_i$ can be written as,
\begin{equation*}
  r_i= \real{h_i}\real{y_i}+ \imag{h_i}\imag{y_i}= |h_i|^2 x_i + |h_i|w_i.
\end{equation*}
The bit-by-bit quantized (i.e., hard) decisions are then defined as,
\begin{equation*}
  \ch_i=\MPi{\sign{r_i}}
\end{equation*}
where $\sign{\cdot}$ denotes the sign of a real number.

More importantly, even though the primary metric of our interest is the BER
performance of the code $\C$, it is mathematically more convenient to obtain
and analyze the list decoding algorithms assuming the probability of codeword
error. Thus, we assume that the list decoding with a given decoding complexity
obtained for the probability of codeword error will also have a good BER
performance. The maximum likelihood (ML) decoder minimizing the probability of
codeword error provides the decision $\cml$ on the most likely transmitted
codeword, i.e.,
\begin{eqnarray}
  \cml&=& \argmin_{\cv\in\C:\ \xv=\MP{\cv}} \norm{\yv -
    \hv\odot\xv}^2\nonumber\\ &=& \argmax_{\cv\in\C} \sum_{i=1}^N \real{y_i
    h_i^\ast x_i^\ast} \overset{\mbox{\footnotesize
      BPSK}}{=}\argmax_{\cv\in\C:\ \xv=\MP{\cv}}\ \rv\cdot\xv \label{eq:10}
\end{eqnarray}
where $\yv$, $\hv$, $\xv$, and $\rv$ denote the $N$-dimensional row vectors of
the received signals $y_i$, the channel coefficients $h_i$, the transmitted
symbols $x_i$, and the reliabilities $r_i$ within one codeword, respectively,
$\norm{\cdot}$ is the Euclidean norm of a vector, $\odot$ is the component-wise
(Hadamard) product of vectors, and the binary operator $\cdot$ is used to
denote the dot-product of vectors. The codewords $\cv\in\C$ used in
\eref{eq:10} to find the maximum or the minimum value of the ML metric are
often referred to as the test codewords. In the following subsection, we
investigate the soft-decision decoding algorithms with low implementation
complexity to replace the computationally demanding ML decoding \eref{eq:10}.

\subsection{List Decoding}

We investigate the list-based decoding algorithms. For simplicity, we assume
binary block codes that are linear and systematic \cite{Lin83}. We note that
whereas the extension of the list-based decoding algorithms to non-systematic
codes is straightforward, the list based decoding of non-linear codes is
complicated by the fact that the list of the test codewords is, in general,
dependent on the received sequence. The decoding (time) complexity $O$ of the
list decoding algorithms can be measured as the list size given by the number
of the test codewords that are examined in the decoding process. Thus, the ML
decoding \eref{eq:10} has the complexity, $O_{\mathrm{ML}}=2^K$, which is
prohibitive for larger values of $K$. Among the practical list-based decoding
algorithms with the acceptable decoding complexity, we investigate the order
statistics decoding (OSD) \cite{Fos95} based list decoding algorithms for
soft-decision decoding of linear binary block codes.

The OSD decoding resumes by reordering the received sequence of reliabilities
as,
\begin{equation}\label{eq:3}
  |\rt_1^\prime|\geq |\rt_2^\prime|\geq \cdots |\rt_N^\prime|
\end{equation}
where the tilde is used to denote the ordering. This ordering of the
reliabilities defines a permutation, $\lambdap$, i.e.,
\begin{equation*}
  \rvt^\prime= \lambdap \sbsb{\rv}= (\rt_1^\prime,\cdots,\rt_N^\prime).
\end{equation*}
The permutation $\lambdap$ corresponds to the generator matrix $\Gvt^\prime=
\lambdap\sbsb{\Gv}$ having the reordered columns. In order to obtain the most
reliable independent positions (MRIPs) for the first $K$ bits in the codeword,
additional swapping of the columns of $\Gvt^\prime$ may have to be used which
corresponds to the permutation $\lambdapp$, and the generator matrix
$\Gvt^\pprime=\lambdapp \sbsb{\Gvt^\prime}$. The matrix $\Gvt^\pprime$ can be
manipulated into a row (or a reduced row) echelon form using the Gauss (or the
Gauss-Jordan) elimination. To simplify the notation, let $\rvt$ and $\Gvt$
denote the reordered sequence of the reliabilities $\rv$ and the reordered
generator matrix $\Gvt$ in a row (or a reduced row) echelon form, respectively,
after employing the permutations $\lambdap$ and $\lambdapp$, to decode the
received sequence $\yv$. Thus, for $i\geq j$, the reordered sequence $\rvt$ has
elements, $|\rt_i|\geq |\rt_j|$, for $i,j=1,\cdots,K$, and for
$i,j=K+1,\cdots,N$.

The complexity of the ML decoding \eref{eq:10} of the received sequence $\yv$
can be reduced by assuming a list of the $L$ test codewords, so that $L\ll
2^K$. Denote such a list of the test codewords of cardinality $L$ generated by
the matrix $\Gvt$ as, $\ES_{L}=\{\ev_0,\ev_2,\cdots,\ev_{L_-1}\}$, and let
$\ev_0=\Zs$ be the all-zero codeword. Then, the list decoding of $\yv$ is
defined as,
\begin{equation}\label{eq:20}
  \cvh= \argmax_{\ev\in\ES_{L}:\ \xv=\MP{\cvh_0\oplus \ev}} \rvt\cdot\xv
\end{equation}
where the systematic part of the codeword $\cvh_0$ is given by the
hard-decision decoded bits at the MRIPs. The decoding step to obtain the
decision $\cvh_0$ is referred to as the $0$-th order OSD reprocessing in
\cite{Fos95}. In addition, due to linearity of $\C$, we have that
$(\cv_0\oplus\ev)\in\C$, and thus, the test codewords $\ev\in\ES_L$ can be also
referred to as the test error patterns in the decoding \eref{eq:20}. Using the
property \eref{eq:5}, we can rewrite the decoding \eref{eq:20} as,
\begin{equation}\label{eq:30}
  \cvh= \argmax_{\ev\in\ES_{L}}\ \rvt\cdot\xvh_0\cdot\MP{\ev}=
  \argmax_{\ev\in\ES_{L}}\ \rvt_0\cdot\MP{\ev}
\end{equation}
where we denoted $\xvh_0=\MP{\cvh_0}$ and $\rvt_0=\rvt\odot\xvh_0$. The system
model employing the list decoding \eref{eq:30} is illustrated in
\fref{Fig:0}. More importantly, as indicated in \fref{Fig:0}, the system model
can be represented as an equivalent channel with the binary vector input $\cv$
and the vector soft-output $\rvt_0$.

\section{List Selection}

The selection of the test error patterns $\ev$ to the list $\ES_{L}$ as well as
the list size $L$ have a dominant effect upon the probability of incorrect
codeword decision by the list decoding. Denote such probability of codeword
error as $\Pe$, and let $\ctx$ be the transmitted codeword.  In \cite{Most09},
the probability $\Pe$ is expanded as,
\begin{eqnarray*}
  \Pe&=& \Prob{\cvh\neq \ctx | \cml\neq \ctx}\Prob{\cml\neq \ctx}\\
  &&+ \Prob{\cvh\neq \ctx | \cml=\ctx}\Prob{\cml=\ctx}
\end{eqnarray*}
where the decision $\cvh$ is obtained by the decoding \eref{eq:30}, and the
condition, $\cml\neq \ctx$, is true provided that the vectors $\cml$ and $\ctx$
differ in at least one component, i.e., $\cml=\ctx$ if and only if all the
components of the vectors are equal. Since, for any list $\ES_{L}$, the
probability, $\Prob{\cvh\neq \ctx | \cml\neq \ctx}=1$, and usually, the
probability, $\Prob{\cml=\ctx}$ is close to $1$, $\Pe$ can be tightly
upper-bounded as,
\begin{equation}\label{eq:40}
  \Pe\leq \Prob{\cml\neq \ctx}+\Prob{\cvh\neq \ctx | \cml=\ctx}.
\end{equation}
The first term on the right hand side of \eref{eq:40} is the codeword error
probability of the ML decoding, and the second term is the conditional codeword
error probability of the list decoding. The probability, $\Prob{\cvh\neq \ctx |
  \cml=\ctx}$, is decreasing with the list size. In the limit of the maximum
list size when the list decoding becomes the ML decoding, the bound
\eref{eq:40} becomes, $\Pe= \Prob{\cml\neq \ctx}$. The bound \eref{eq:40} is
particularly useful to analyze the performance of the list decoding
\eref{eq:30}. However, in order to construct the list of the test error
patterns, we consider the following expansion of the probability $\Pe$, i.e.,
\begin{eqnarray*}
  \Pe&=& \Prob{\cvh\neq \ctx| (\ctx\oplus\cvh_0)\in\ES_{L}}
  \Prob{(\ctx\oplus\cvh_0)\in\ES_{L}}\\
  &&+  \Prob{\cvh\neq \ctx| (\ctx\oplus\cvh_0)\not\in\ES_{L}}
  \Prob{(\ctx\oplus\cvh_0)\not\in\ES_{L}}\\
  &=& 1- \underbrace{\Prob{ \cvh= \ctx | (\ctx\oplus\cvh_0)\in\ES_{L}}}_\PI
  \underbrace{\Prob{(\ctx\oplus\cvh_0)\in\ES_{L}}}_\PII.
\end{eqnarray*}
Using \eref{eq:20} and \eref{eq:30}, the probability $\PI$ that the list
decoding \eref{eq:30} selects the transmitted codeword provided that such
codeword is in the list (more precisely, provided that the error pattern
$\ctx\oplus\cvh_0$ is in the list) can be expressed as,
\begin{equation}\label{eq:50}
  \PI= \Prob{ \rvt\cdot\MP{\ctx\oplus\cvh_0}\geq \rvt_0\cdot\MP{\ev},
    \forall\ev\in\ES_{L}}.
\end{equation}
The probability \eref{eq:50} decreases with the list size, and, in the limit of
the maximum list size $L=2^K$, $\PI=1-\Pe$. On the other hand, the probability
$\PII$ that the transmitted codeword is in the decoding list increases with the
list size, and $\PII=1$, for $L=2^K$.

Since the coding $\C$ and the communication channel are linear, then, without
loss of generality, we can assume that the all-zero codeword, $\ctx=\Zs$, is
transmitted. Consequently, given the list decoding complexity $L$, the optimum
list $\ES_L^\ast$ minimizing the probability $\Pe$ is constructed as,
\begin{equation}\label{eq:60}
  \ES_L^\ast= \argmax_{\ES:\ |\ES|=L} \Prob{\cvh=\Zs|\cvh_0\in\ES}
  \Prob{\cvh_0\in\ES}
\end{equation}
where $|\ES|$ is the cardinality of the test list $\ES$, and the hard-decision
codeword $\cvh_0\in\C$ represents the error pattern observed at the receiver
after transmission of the codeword $\ctx=\Zs$. For a given list of the error
patterns $\ES$ in \eref{eq:60}, and for the system model in Section II with
asymptotically large SNR, the probability $\PI= \Prob{\cvh=\Zs|\cvh_0\in\ES}$
is dominated by the error events corresponding to the error patterns with the
smallest Hamming distances. Since the error patterns are also codewords of
$\C$, the smallest Hamming distance between any two error patterns in the list
$\ES$ is at least $\dmin$. Assuming that the search in \eref{eq:60} is
constrained to the lists $\ES$ having the minimum Hamming distance between any
two error patterns given by $\dmin$, the probability $\PI$ is approximately
constant for all the lists $\ES$, and we can consider the suboptimum list
construction,
\begin{equation}\label{eq:61}
  \ES_L= \argmax_{\ES:\ |\ES|=L} \Prob{\cvh_0\in\ES}.
\end{equation}
The list construction \eref{eq:61} is recursive in its nature, since the list
$\ES$ maximizing \eref{eq:61} consists of the $L$ most probable error
patterns. However, in order to achieve a small probability of decoding error
$\Pe$ and approach the probability of decoding error, $\Prob{\cml\neq \ctx}$,
of the ML decoding, the list size $L$ must be large. We can obtain a practical
list construction by assuming the $L$ sufficiently probable error patterns
rather than assuming the $L$ most likely error patterns. We restate Theorem 1
and Theorem 2 in \cite{Fos95} to obtain the likely error patterns and to define
the practical list decoding algorithms.

Denote as $\Prm(i_{1},i_{2},\cdots,i_{n})$ the $n$-th order joint probability
of bit errors at bit positions $1\leq i_{1}< i_{2}< \cdots <i_{n}\leq N$ in the
received codeword after the ordering $\lambdap$ and $\lambdapp$ and before the
decoding. Since the test error pattern $\ev$ is a codeword of $\C$, the
probability $\Prm(i_{1},i_{2},\cdots,i_{n})$, for $i_n\leq K$, is equal to the
probability $\Prob{\ev=\cvh_0}$ assuming that the $n$ bit errors occurred during
the transmission corresponding to the positions (after the ordering)
$i_1,i_2,\cdots,i_n$. We have the following lemma.
\begin{lemma}\label{lm:1}
  For any bit positions $\Ic_1\subseteq\Ic\subseteq \{1,2,\cdots,N\}$,
  \begin{equation*}
    \Prm(\Ic)\leq \Prm(\Ic_1).
  \end{equation*}
\end{lemma}
\begin{proof}
  The lemma is proved by noting that $\Prm(\Ic)=\Prm(\Ic_1,\Ic\backslash\Ic_1)
  = \Prm(\Ic\backslash\Ic_1| \Ic_1)\Prm(\Ic_1)\leq \min\{\Prm(\Ic_1),
  \Prm(\Ic\backslash\Ic_1| \Ic_1)\}\leq \Prm(\Ic_1)$ where $\Ic\backslash
  \Ic_1$ denotes the difference of the two sets.
\end{proof}
Using \lmref{lm:1}, we can show, for example, that, $\Prm(i,j)\leq\Prm(i)$ and
$\Prm(i,j)\leq \Prm(j)$. We can now restate Theorem 1 and Theorem 2 in
\cite{Fos95} as follows.
\begin{theorem}\label{th:1}
  Assume bit positions $1\leq i<j<k\leq N$, and let the corresponding
  reliabilities be $|\rt_i|\geq |\rt_j|\geq |\rt_k|$. Then, the bit error
  probabilities,
  \begin{eqnarray*}
    \Prm(i)&\leq& \Prm(j)\\
    \Prm(i,j)&\leq& \Prm(i,k).
  \end{eqnarray*}
\end{theorem}
\begin{proof}
  Without loss of generality, we assume that the symbols $x_i=-1$, $x_j=-1$ and
  $x_k=-1$ have been transmitted. Then, before the decoding, the received bits
  would be decided erroneously if the reliabilities $\rt_i>0$, $\rt_j>0$, and
  $\rt_k>0$. Conditioned on the transmitted symbols, let $f(\cdot)$ denote the
  conditional PDF of the ordered reliabilities $\rt_i$, $\rt_j$ and $\rt_k$.

  Consider first the inequality $\Prm(i)\leq \Prm(j)$. Since, for $\rt_i>0$,
  $f(\rt_i)<f(-\rt_i)$, using $f(\rt_i,\rt_j)=f(\rt_i|\rt_j)f(\rt_j)$, we can
  show that, for $\rt_i>0$ and any $\rt_j$, $f(\rt_i,\rt_j)<
  f(-\rt_i,\rt_j)$. Similarly, using $f(-\rt_i,\rt_j)=
  f(\rt_j|-\rt_i)f(-\rt_i)$, we can show that, for $\rt_j>0$ and any $\rt_i$,
  $f(-\rt_i,\rt_j)< f(-\rt_i,-\rt_j)$. Then, the probability of error for bits
  $i$ and $j$, respectively, is,
  \begin{eqnarray*}
    \Prm(i)&=& \int_0^\infty \int_{-\rt_i}^{\rt_i} f(\rt_i,\rt_j)
    \df\rt_j\df\rt_i\\ &=& \int_0^\infty \int_0^{\rt_i} f(\rt_i,\rt_j) \df\rt_j
    \df\rt_i+ \int_0^\infty \int_{-\rt_i}^0 f(\rt_i,\rt_j) \df\rt_j\df\rt_i\\
    \Prm(j)&=& \int_{-\infty}^\infty \int_0^{|\rt_i|} f(\rt_i,\rt_j) \df\rt_j
    \df\rt_i \\ &=& \int_0^\infty \int_0^{\rt_i} f(\rt_i,\rt_j) \df\rt_j\df\rt_i+
    \int_0^\infty\int_{-\rt_i}^0 f(-\rt_i,-\rt_j) \df\rt_j\df\rt_i
  \end{eqnarray*}
  and thus, $\Prm(i)\leq \Prm(j)$.

  The second inequality, $\Prm(i,j)\leq \Prm(i,k)$, can be proved by assuming
  conditioning, $\Prm(i,j)=\Prm(j|i)\Prm(i)$, $\Prm(i,k)=\Prm(k|i)\Prm(i)$, and
  $f(\rt_i,\rt_j,\rt_k)= f(\rt_j,\rt_k|\rt_i) f(\rt_i)$, and by using
  inequality $\Prm(i)\leq \Prm(j)$, and following the steps in the first part
  of the theorem proof.
\end{proof}

\section{List Decoding Algorithms}

Using Theorem 1 and Theorem 2 in \cite{Fos95}, the original OSD assumes the
following list of error patterns,
\begin{equation}\label{eq:15}
  \ES_L=\{\ev\Gv:\ 0\leq\wH{\ev}\leq I,\ \ev\in\Z_2^K\}
\end{equation}
where $I$ is the so-called reprocessing order of the OSD, and $\wH{\ev}$ is the
Hamming weight of the vector $\ev$. The list \eref{eq:15} uses a
$K$-dimensional sphere of radius $I$ defined about the origin
$\Zs=(0,\cdots,0)$ in $\Z_2^K$. The decoding complexity for the list
\eref{eq:15} is $L=\sum_{l=0}^I \binom{K}{l}$ where $l$ is referred to as the
phase of the order $I$ reprocessing in \cite{Fos95}. Assuming an AWGN channel,
the recommended reprocessing order is $I=\lceil \dmin/4\rceil\ll K$ where
$\lceil\cdot\rceil$ is the ceiling function. Since the OSD algorithm may become
too complex for larger values of $I$ and $K$, a stopping criterion for
searching the list $\ES_L$ was developed in \cite{Fos96}.

We can identify the following inefficiencies of the original OSD
algorithm. First, provided that no stopping nor skipping rules for searching
the list of the test error patterns are used, once the MRIPs are found, the
ordering of bits within the MRIPs according to their reliabilities becomes
irrelevant.  Second, whereas the BER performance of the OSD is modestly
improving with the reprocessing order $I$, the complexity of the OSD increases
rapidly with $I$ \cite{Fos96}. Thus, for given $K$, the maximum value of $I$ is
limited by the allowable OSD complexity to achieve a certain target BER
performance. We can address the inefficiencies of the original OSD by more
carefully exploiting the properties of the probability of bit errors given by
\lmref{lm:1} and \thref{th:1}. Hence, our aim is to construct a well-defined
list of the test error patterns \emph{without} considering the stopping and the
skipping criteria to search this list.

Recall that the error patterns can be uniquely specified by bits in the MRIPs
whereas the bits of the error patterns outside the MRIPs are obtained using the
parity check matrix. In order to design a list of the test error patterns
independently of the particular generator matrix of the code as well as
independently of the particular received sequence, we consider only the bit
errors within the MRIPs. Thus, we can assume that, for all error patterns, the
bit errors outside the MRIPs affect the value of the metric in \eref{eq:30}
equally. More importantly, in order to improve the list decoding complexity and
the BER performance trade-off, we consider partitioning of the MRIPs into
disjoint segments. This decoding strategy employing the segments of the MRIPs
is investigated next.

\subsection{Segmentation-Based OSD}

Assuming $Q$ disjoint segments of the MRIPs, the error pattern $\ev$
corresponding to the $K$ MRIPs can be expressed as a concatenation of the $Q$
error patterns $\ev^{(q)}$ of length $K_q$ bits, $q=1,\cdots,Q$, i.e.,
\begin{equation*}
  \ev=(\ev^{(1)},\cdots,\ev^{(Q)})\in\Z_2^K
\end{equation*}
so that $\sum_{q=1}^Q K_q= K$, and $\wH{\ev}=\wH{\ev^{(1)}}+ \cdots+
\wH{\ev^{(Q)}}$. As indicated by \lmref{lm:1} and \thref{th:1}, more likely
error patterns have smaller Hamming weights and they correct the bit positions
with smaller reliabilities. In addition, the decoding complexity given by the
total number of error patterns in the list should grow linearly with the number
of segments $Q$. Consequently, for a small number of segments $Q$, it is
expected that a good decoding strategy is to decode each segment independently,
and then, the final decision is obtained by selecting the best error
(correcting) pattern from each of the segments decodings. In this paper, we
refine this strategy for $Q=2$ segments as a generalization of the conventional
OSD having only $Q=1$ segment.

Assuming that the two segments of the MRIPs are decoded independently, the list
of error patterns can be written as,
\begin{equation}\label{eq:18}
  \ES_L=\ESi\cup \ESii
\end{equation}
where $\ESi$ and $\ESii$ are the lists of error patterns corresponding to the
list decoding of the first segment and of the second segment, respectively, and
$L=L_1+L_2$. Obviously, fewer errors, and thus, fewer error patterns can be
assumed in the shorter segments with larger reliabilities of the received
bits. Similarly to the conventional OSD having one segment, for both MRIPs
segments, we assume all the error patterns up to the maximum Hamming weight
$I_q$, $q=1,2$. Then, the lists of error patterns in \eref{eq:18} can be
defined as,
\begin{equation}\label{eq:16}
  \begin{array}{rcl}
  \ESi&=&\{(\ev,\Zs)\Gv:\ 0\leq\wH{\ev}\leq I_1,\ \ev\in\Z_2^{K_1}\} \\
  \ESii&=&\{(\Zs,\ev)\Gv:\ 0\leq\wH{\ev}\leq I_2,\ \ev\in\Z_2^{K_2}\}.
  \end{array}
\end{equation}
The decoding complexity of the segmentation-based OSD with the lists of error
patterns defined in \eref{eq:16} is,
\begin{equation*}
  L=\sum_{l=0}^{I_1} \binom{K_1}{l}+ \sum_{l=0}^{I_2} \binom{K_2}{l}
\end{equation*}
where $K=K_1+K_2$, and we assume $I_1\ll K_1$ and $I_2\ll K_2$. 

Recall that the original OSD, denoted as $\OSD(I)$, has one segment of length
$K$ bits, and that the maximum number of bit errors assumed in this segment is
$I$.  The segmentation-based OSD is denoted as $\OSD(I_1,I_2)$, and it is
parametrized by the segment length $K_1,$ and $K_2$, and the maximum number of
errors $I_1$ and $I_2$, respectively. The segment sizes $K_1$ and $K_2$ are
chosen empirically to minimize the BER for a given decoding complexity and for
a class of codes under consideration. In particular, for systematic block codes
of block length $N<128$ and of rate $R\geq 1/2$, we found that the recommended
length of the first segment is,
\begin{equation*}
  K_1\approx 0.35 K
\end{equation*}
so that the second segment length is $K_2=K-K_1$. The maximum number of bit
errors $I_1$ and $I_2$ in the two segments are selected to fine-tune the BER
performance and the decoding complexity trade-off. For instance, we can obtain
the list decoding schemes having the BER performance as well as the decoding
complexity between those corresponding to the original decoding schemes
$\OSD(I)$ and $\OSD(I+1)$.

Finally, we note that it is straightforward to develop the skipping criteria
for efficient searching of the list of error patterns in the OSD-based decoding
schemes. In particular, one can consider the Hamming distances for one or more
segments of the MRIPs between the received hard decisions (before the decoding)
and the temporary decisions obtained using the test error patterns from the
list. If any or all of the Hamming distances are above given thresholds, the
test error pattern can be discarded without re-encoding and calculating the
corresponding Euclidean distance. For the $Q=2$ segments OSD, our empirical
results indicate that the thresholds for the first and the second segments
should be $0.35\,\dmin$ and $\dmin$, respectively.

\subsection{Partial-Order Statistics Decoding}

The Gauss (or the Gauss-Jordan) elimination employed in the OSD-based decoding
algorithms represents a significant portion of the overall implementation
complexity.  A new row (or a reduced row) echelon form of the generator matrix
must be obtained after every permutation $\lambdapp$ until the MRIPs are
found. Hence, we can devise a partial-order statistics decoding (POSD) that
completely avoids the Gauss elimination, and thus, it further reduces the
decoding complexity of the OSD-based decoding. The main idea of the POSD is to
order only the first $K$ received bits according to their reliabilities, so
that the generator matrix remains in its systematic form. The ordering of the
first $K$ received bits in the descending order can improve the coding gain of
the segmentation-based OSD. Assuming $Q=2$ segments, we use the notation
$\POSD(I_1,I_2)$. The parameters $K_1$, $K_2$, $I_1$ and $I_2$ of
$\POSD(I_1,I_2)$ can be optimized similarly as in the case of $\OSD(I_1,I_2)$
to fine-tune the BER performance versus the implementation complexity. On the
other hand, we will show in Section V that the partial ordering (i.e., the
ordering of the first $K$ out of $N$ received bits) is irrelevant for the OSD
decoding having one segment of the MRIPs and using the list of error patterns
\eref{eq:15}. In this case, the $\POSD(I)$ decoding can be referred to as the
input-sphere decoding $\ISD(I)$.

\subsection{Implementation Complexity}

We compare the number of binary operations (BOPS) and the number of floating
point operations (FLOPS) required to execute the decoding algorithms proposed
in this paper. Assuming a $(N,K,\dmin)$ code, the complexity of the OSD and the
POSD are given in \tbref{Tab:1} and \tbref{Tab:2}. The implementation
complexity expressions in \tbref{Tab:1} for $\OSD(I)$ are from the reference
\cite{Fos95}. For example, the OSD decoding of the BCH code $(128,64,22)$
requires at least $1152$ FLOPS and $528448$ BOPS to find the MRIPs and to
obtain the corresponding equivalent generator matrix in a row echelon form. All
this complexity can be completely avoided by assuming the partial ordering in
the POSD decoding. The number of the test error patterns is $L=2080$ for
$\OSD(2)$, and $L=1177$ for $\OSD(2,2)$ with $K_1=21$ and $K_2=43$ whereas the
coding gain of $\OSD(2)$ can be only slightly better than the coding gain of
$\OSD(2,2)$; see, for example, \fref{Fig:1}. Hence, the overall complexity of
the OSD-based schemes can be substantially reduced by avoiding the Gauss
(Gauss-Jordan) elimination.

\section{Performance Analysis}

Recall that we assume a memoryless communication channel as described in
Section II. We derive the probability $\Prob{\cvh_0\in\ES_L}$ in \eref{eq:61}
that the error pattern $\cvh_0$ observed at the receiver after transmission of
the codeword $\ctx=\Zs$ is an element of the chosen decoding list $\ES_L$. The
derivation relies on the following generalization of Lemma 3 in \cite{Fos95}.
\begin{lemma}\label{lm:2}
  For any ordering of the $N$ received bits, consider the $I$ bit positions
  $\Ic\subseteq \{1,2,\cdots,N\}$, and the $\binom{I}{I_1}$ subsets
  $\Ic_1\subseteq \Ic$ of $I_1\leq I\leq N$ bit positions. Then, the total
  probability of the $I_1$ bit errors within the $I$ bits can be calculated as,
  \begin{equation*}
    \sum_{\Ic_1:\ |\Ic_1|=I_1} \Prm(\Ic_1)= \binom{I}{I_1}\, p_0^{I_1}
  \end{equation*}
  where $p_0$ is the probability of bit error corresponding to the bit
  positions $\Ic$.
\end{lemma}
\begin{proof}
  The ordering of the chosen $I$ bits given by the set $\Ic$ is irrelevant
  since \emph{all} subsets $\Ic_1$ of $I_1$ errors within the $I$ bits $\Ic$
  are considered. Consequently, the bit errors in the set $\Ic$ can be
  considered to be independent having the equal probability denoted as $p_0$.
\end{proof}
Using \lmref{lm:2}, we observe that the lists of error patterns \eref{eq:15}
and \eref{eq:16} are constructed, so that the ordering of bits within the
segments is irrelevant. Then, the bit errors in a given segment can be
considered to be conditionally independent. This observation is formulated as
the following corollary of \lmref{lm:2}.
\begin{corollary}\label{cl:1}
  For the $\OSD(I)$ and the list of error patterns \eref{eq:15}, the bit errors
  in the MRIPs can be considered as conditionally independent. Similarly, for
  the $\POSD(I_1,I_2)$ and the list of error patterns \eref{eq:16}, the bit
  errors in the two segments can be considered as conditionally independent.
\end{corollary}
Thus, the bit errors in \clref{cl:1} are independent conditioned on the
particular segment being considered as shown next.

Let $\Prm_0$ be the bit error probability of the MRIPs for the $\OSD(I)$
decoding. Similarly, let $\Prm_1$ and $\Prm_2$ be the bit error probabilities
in the first and the second segments of the $\OSD(I_1,I_2)$ decoding,
respectively. Denote the auxiliary variables, $v_1=|\rt_{K_1}|$,
$v_2=|\rt_{K_1+1}|$, and $v_3=|\rt_{K+1}|$ of the order statistics \eref{eq:3},
and let $u\equiv|r_i|$, $i=1,2,\cdots,K$. Hence, always, $v_1\geq v_2$, and,
for simplicity, ignoring the second permutation $\lambdapp$, also, $v_2\geq
v_3$. The probability of bit error $\Prm_0$ for the MRIPs is calculated as,
\begin{equation*}
  \Prm_0= \EE{u}{\int_0^u \frac{f_{v_3}(v)}{1-F_u(v)}\df v}
\end{equation*}
where $\EE{u}{\cdot}$ denotes the expectation w.r.t. (with respect to) $u$,
$f_{v_3}(v)$ is the PDF of the $(K+1)$-th order statistic in \eref{eq:3}, and
$F_u(v)$ is the cumulative distribution function (CDF) of the magnitude (the
absolute value) of the reliability of the received bits (before
ordering). Similarly, the probability of bit error $\Prm_1$ for the first
segment is calculated as,
\begin{equation*}
  \Prm_1= \EE{u}{\int_0^u \frac{f_{v_2}(v)}{1-F_u(v)}\df v}
\end{equation*}
where $f_{v_2}(v)$ is the PDF of the $(K_1+1)$-th order statistic in
\eref{eq:3}. The probability of bit error $\Prm_2$ for the second segment is
calculated as,
\begin{equation*}
  \Prm_2= \EE{u}{\int_0^u \!\int_u^\infty \frac{f_{v_1}(v)
      f_{v_3}(v^\prime)}{(F_u(v)- F_u(v^\prime))(1-F_{v_1}(v^\prime))}
    \df v \df v^\prime}
\end{equation*}
where $f_{v_1}(v)$ and $F_{v_1}(v^\prime)$ is the PDF and the CDF of the
$K_1$-th order statistic in \eref{eq:3}, respectively. The values of the
probabilities $\Prm_0$, $\Prm_1$ and $\Prm_2$ have to be evaluated by numerical
integration. Finally, we use \lmref{lm:2} and substitute the probabilities
$\Prm_0$, $\Prm_1$ and $\Prm_2$ for $p_0$ to calculate the probability
$\Prob{\cvh_0\in\ES_L}$ of the error patterns in the list $\ES_L$.

\section{Numerical Examples}

We use computer simulations to compare the BER performances of the proposed
soft-decision decoding schemes. All the block codes considered are linear and
systematic.

The BER of the $(31,16,7)$ BCH code over an AWGN channel is shown in
\fref{Fig:2} assuming $\ISD(2)$ and $\ISD(3)$ with $K=16$ having $137$ and
$697$ test error patterns, respectively, and assuming $\POSD(1,3)$ and
$\POSD(2,3)$ with $K_1=6$ and $K_2=10$ having $183$ and $198$ test error
patterns, respectively. We observe that $\POSD(1,3)$ achieves the same BER as
$\ISD(3)$ while using much less error patterns which represents the gain of the
ordering of the received information bits into two segments. At the BER of
$10^{-4}$, $\POSD(1,3)$ outperforms $\ISD(2)$ by $1.1$ dB using approximately
$50\%$ more test error patterns. Thus, the $\POSD(1,3)$ decoding provides $2.3$
dB coding gain with the small implementation complexity at the expense of $2$
dB loss compared to the ML decoding.

\fref{Fig:4} shows the BER of the $(63,45,14)$ BCH code over an AWGN
channel. The number of test error patterns for the $\ISD(2)$, $\ISD(3)$,
$\POSD(1,3)$ and $\OSD(2)$ decodings are $1036$, $15226$, $5503$ and $1036$,
respectively. We observe from \fref{Fig:4} that $\ISD(3)$ has the same BER as
$\POSD(1,3)$ with two segments of $K_1=13$ and $K_2=32$ bits. However,
especially for the high rate codes (i.e. of rates greater than $1/2$), one has
to also consider the complexity of the Gauss elimination to obtain the row
echelon form of the generator matrix for the OSD. For example, the Gauss
elimination for the $(63,45,14)$ code requires approximately $20,400$ BOPS;
cf. \tbref{Tab:1}.

The BER of the $(128,64,22)$ BCH code over an AWGN channel is shown in
\fref{Fig:1} assuming $\OSD(1)$ and $\OSD(2)$ with $K=64$, and assuming
$\OSD(2,2)$ with $K_1=21$ and $K_2=43$. The number of test error patterns for
the $\OSD(1)$, $\OSD(2)$ and $\OSD(2,2)$ decodings are $64$, $2081$ and $1179$.
A truncated union bound of the BER in \fref{Fig:1} is used to indicate the ML
performance \cite[Ch. 10]{Ben99}. We observe that both $\OSD(2)$ and
$\OSD(2,2)$ have the same BER performance for the BER values larger than
$10^{-3}$, and $OSD(2)$ outperforms $\OSD(2,2)$ by at most $0.5$ dB for the
small values of the SNR. Our numerical results indicate that, in general,
$\OSD(2,2)$ decoding can achieve approximately the same BER as $\OSD(2)$ for
small to medium SNR while using about $50\%$ less error patterns. Thus, a
slightly smaller coding gain (less than $0.5$dB) of $\OSD(2,2)$ in comparison
with $\OSD(2)$ at larger values of the SNR is well-compensated for by the
reduced decoding complexity. More importantly, $\OSD(2,2)$ can trade-off the
BER performance and the decoding complexity between those provided by $\OSD(1)$
and $\OSD(2)$, especially at larger values of SNR.

The BER of the $(31,16,7)$ BCH code over a fast Rayleigh fading channel is
shown in \fref{Fig:3}. We assume the same decoding schemes as in \fref{Fig:2}.
The $\POSD(1,3)$ decoding with $183$ error patterns achieves the coding gain of
$17$ dB over an uncoded system, the coding gain of $4$ dB over $\ISD(2)$ with
$137$ error patterns, and it has the same BER as $\OSD(3)$ with $697$ error
patterns. The BER of the high rate BCH code $(64,57,4)$ over a fast Rayleigh
channel is shown in \fref{Fig:9}. In this case, the number of test error
patterns for the $\ISD(2)$, $\ISD(3)$, $\POSD(2,3)$ and $\OSD(2)$ decoding is
$1654$, $30914$, $8685$ and $1654$, respectively. We observe that, for small to
medium SNR, $\POSD(2,3)$ which does not require the Gauss elimination
(corresponding to approximately $3,000$ BOPS) outperforms $\OSD(2)$ by $1$dB
whereas, for large SNR values, these two decoding schemes achieve approximately
the same BER performance.

\section{Conclusions}

Low-complexity soft-decision decoding techniques employing a list of the test
error patterns for linear binary block codes of small to medium block length
were investigated. The optimum and suboptimum construction of the list of
error patterns was developed. Some properties of the joint probability of error
of the received bits after ordering were derived. The original OSD algorithm
was generalized by assuming a segmentation of the MRIPs. The segmentation of
the MRIPs was shown to overcome several drawbacks of the original OSD and to
enable flexibility for devising new decoding strategies. The decoding
complexity of the OSD-based decoding algorithms was reduced further by avoiding
the Gauss (or the Gauss Jordan) elimination using the partial ordering of the
received bits in the POSD decoding. The performance analysis was concerned with
the problem of finding the probability of the test error patterns contained in
the decoding list. The BER performance and the decoding complexity of the
proposed decoding techniques were compared by extensive computer
simulations. Numerical examples demonstrated excellent flexibility of the
proposed decoding schemes to trade-off the BER performance and the decoding
complexity. In some cases, both the BER performance as well as the decoding
complexity of the segmentation-based OSD were found to be improved compared to
the original OSD.

\section*{Appendix}

We derive the probabilities $\Prm_0$, $\Prm_1$ and $\Prm_2$ in Section V.
Without loss of generality, we assume that the all-ones codeword was
transmitted, i.e., $x_i=-1$ for $\forall i$. Then, after ordering, the $i$-th
received bit, $i=1,2,\cdots,N$, is in error, provided that $\rt_i>0$. The
probability of bit error $\Prm_0$ for the MRIPs is obtained as,
\begin{equation*}
  \Prm_0= \int_0^\infty\!\int_0^\infty f_u(u| u\geq v_3) f_{v_3}(v_3) \df v_3
  \df u
\end{equation*}
where the conditional PDF \cite{Papoulis},
\begin{equation*}
  f_u(u| u\geq v_3)= \left\{\begin{array}{cc} \frac{f_u(u)}{1-F_u(v_3)}& u\geq
      v_3\\ 0 & u<v_3\end{array}\right.
\end{equation*}
and $f_u(u)$ and $F_u(v_3)$ are the PDF and the CDF of the reliability of the
received bits, respectively, so that,
\begin{equation*}
  \Prm_0= \int_0^\infty\! f_u(u) \int_0^u \frac{f_{v_3}(v_3)}{1-F_u(v_3)} 
  \df v_3 \df u.
\end{equation*}

Similarly, the probability of bit error $\Prm_1$ for the first segment is
calculated as,
\begin{eqnarray*}
  \Prm_1&=& \int_0^\infty \int_0^\infty f_u(u| u\geq v_2) f_{v_2}(v_2)
  \df v_2 \df u \\
  &=& \int_0^\infty\! f_u(u) \int_0^u \frac{f_{v_2}(v_2)}{1-F_u(v_2)} 
  \df v_2 \df u.
\end{eqnarray*}

The probability of bit error $\Prm_2$ for the second segment is calculated as,
\begin{equation*}
  \Prm_2= \int_0^\infty \int_0^\infty\!\int_0^\infty 
  f_u(u| v_1\geq u\geq v_3) f_{v_1,v_3}(v_1,v_3) \df v_1 \df v_3 \df u
\end{equation*}
where the conditional PDF,
\begin{equation*}
  f_u(u| v_1\geq u\geq v_3)= \left\{\begin{array}{cc} \frac{f_u(u)}{F_u(v_1)-
        F_u(v_3)}& v_1\geq u\geq v_3\\ 0& \mbox{otherwise}\end{array}\right.
\end{equation*}
and the joint PDF of the order statistics $v_1\geq v_3$ is,
\begin{equation*}
  f_{v_1,v_3}(v_1,v_3)= \left\{\begin{array}{cc} \frac{f_{v_1}(v_1)}{
        1-F_{v_1}(v_3)} f_{v_3}(v_3)& v_1\geq v_3\\
      0& v_1<v_3\end{array}\right.
\end{equation*}
and thus,
\begin{equation*}
  \Prm_2= \int_0^\infty f_u(u) \int_0^u \!\int_u^\infty \frac{f_{v_1}(v)
    f_{v_3}(v^\prime)}{(F_u(v)- F_u(v^\prime))(1-F_{v_1}(v^\prime))}
  \df v \df v^\prime \df u.
\end{equation*}

\bibliographystyle{IEEEbib}

\newpage

\begin{table}[t!]
  \begin{center}  
    \caption{Implementation Complexity of the OSD and the POSD}
    \label{Tab:1}
    \begin{tabular}{|c|c|}
      \hline\multicolumn{2}{|c|}{$\OSD(I_1)$ and $\OSD(I_1,I_2)$}\\\hline      
      operation& complexity\\\hline
      $\rv$ & $2N$ FLOPS\\
      $\rvt^\prime$ & $N\log_2(N)$ FLOPS\\
      Gauss el. $\Gv^\prime$& $N\min(K,N-K)^2$ BOPS\\
      $\rvt^\pprime$& $K+K(N-K)$ BOPS\\\hline
      \hline \multicolumn{2}{|c|}{$\POSD(I_1)\equiv\ISD(I_1)$}\\\hline
      operation& complexity\\\hline
      $\rv$ & $2N$ FLOPS\\
      $\rvt^\prime$ & $0$ BOPS\\\hline
      \hline\multicolumn{2}{|c|}{$\POSD(I_1,I_2)$}\\\hline 
      operation& complexity\\\hline
      $\rv$ & $2N$ FLOPS\\
      $\rvt^\prime$ & $K\log_2(K)$ FLOPS\\\hline
    \end{tabular}      
  \end{center}
\end{table}

\begin{table}[t!]
  \begin{center}  
    \caption{Decoding List Sizes for the OSD and the POSD}
    \label{Tab:2}
    \begin{tabular}{|c|c|}\hline
      $\OSD(I)$ \B& $\sum_{l=0}^{I}\binom{K}{l}$ \\\hline
      $\OSD(I_1,I_2)$ \B& $\sum_{l=0}^{I_1}\binom{K_1}{l}+
      \sum_{l=0}^{I_2}\binom{K_2}{l}$  \\\hline
      $\POSD(I)\equiv\ISD(I)$\B& $\sum_{l=0}^{I}\binom{K}{l}$\\\hline
      $\POSD(I_1,I_2)$ \B& $\sum_{l=0}^{I_1}\binom{K_1}{l}+
      \sum_{l=0}^{I_2}\binom{K_2}{l}$ \\\hline
    \end{tabular}      
  \end{center}
\end{table}

\begin{figure}[p]
  \begin{center}
    \includegraphics[scale=1.3]{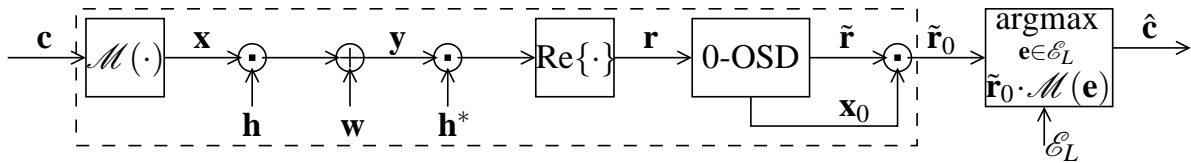}
    \caption{The system model and an equivalent vector channel with the binary
      vector input $\cv$ and the vector soft-output $\rvt_0$.}
    \label{Fig:0}
  \end{center}
\end{figure}

\begin{figure}[p]
  \begin{center}
    \includegraphics[scale=1.0]{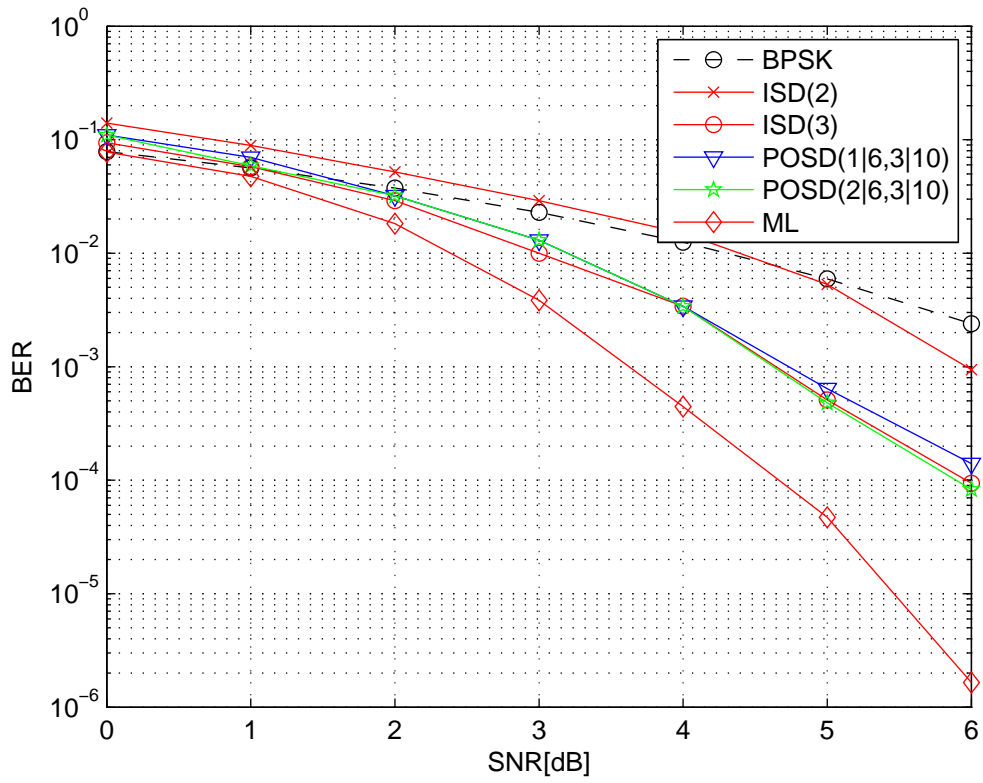}
    \caption{The BER of the $(31,16,7)$ BCH code over an AWGN channel.}
    \label{Fig:2}
  \end{center}
\end{figure}

\begin{figure}[p]
  \begin{center}
    \includegraphics[scale=1.0]{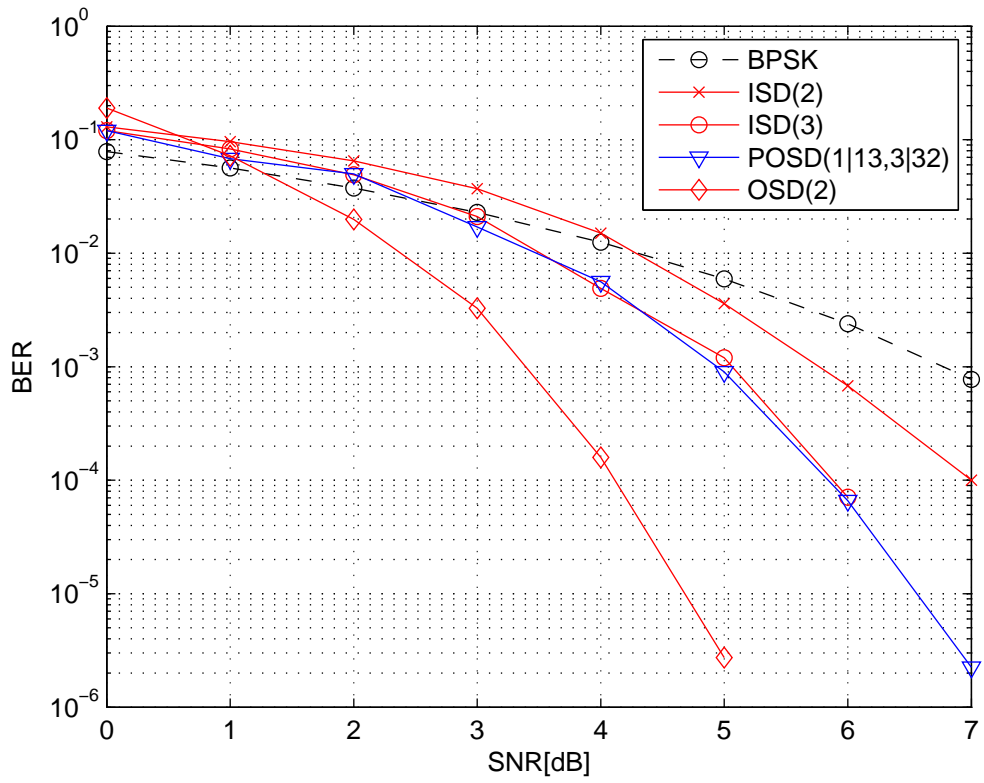}
    \caption{The BER of the $(63,45,14)$ BCH code over an AWGN channel.}
    \label{Fig:4}
  \end{center}
\end{figure}

\begin{figure}[p]
  \begin{center}
    \includegraphics[scale=1.0]{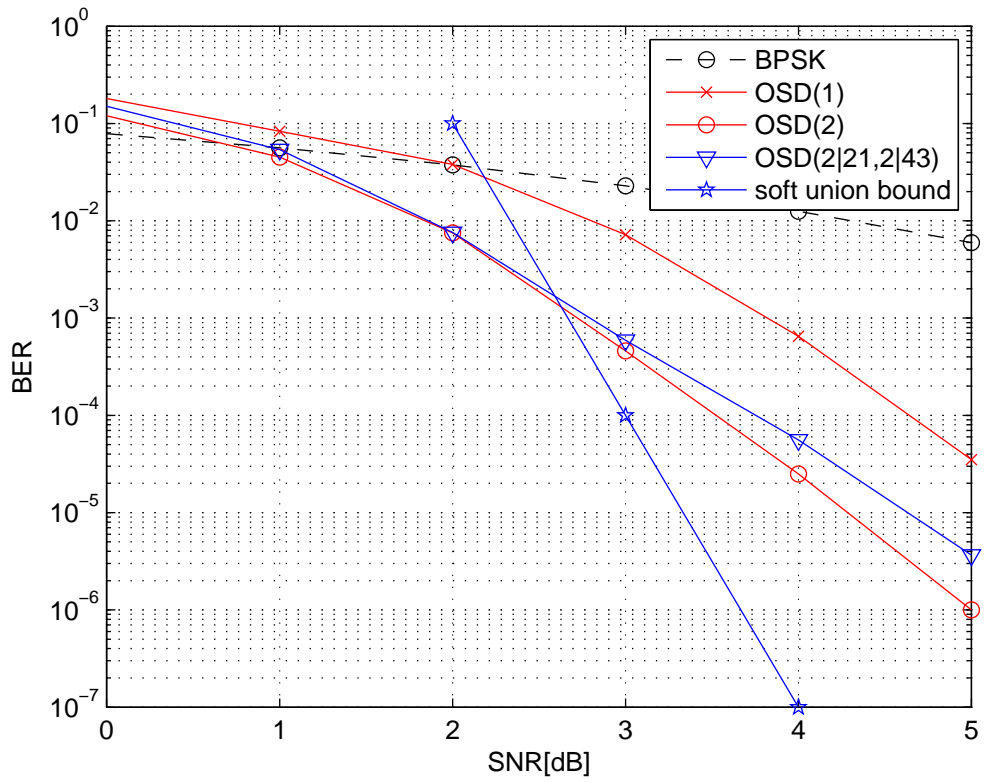}
    \caption{The BER of the $(128,64,22)$ BCH code over an AWGN channel.}
    \label{Fig:1}
  \end{center}
\end{figure}

\begin{figure}[p]
  \begin{center}
    \includegraphics[scale=0.9]{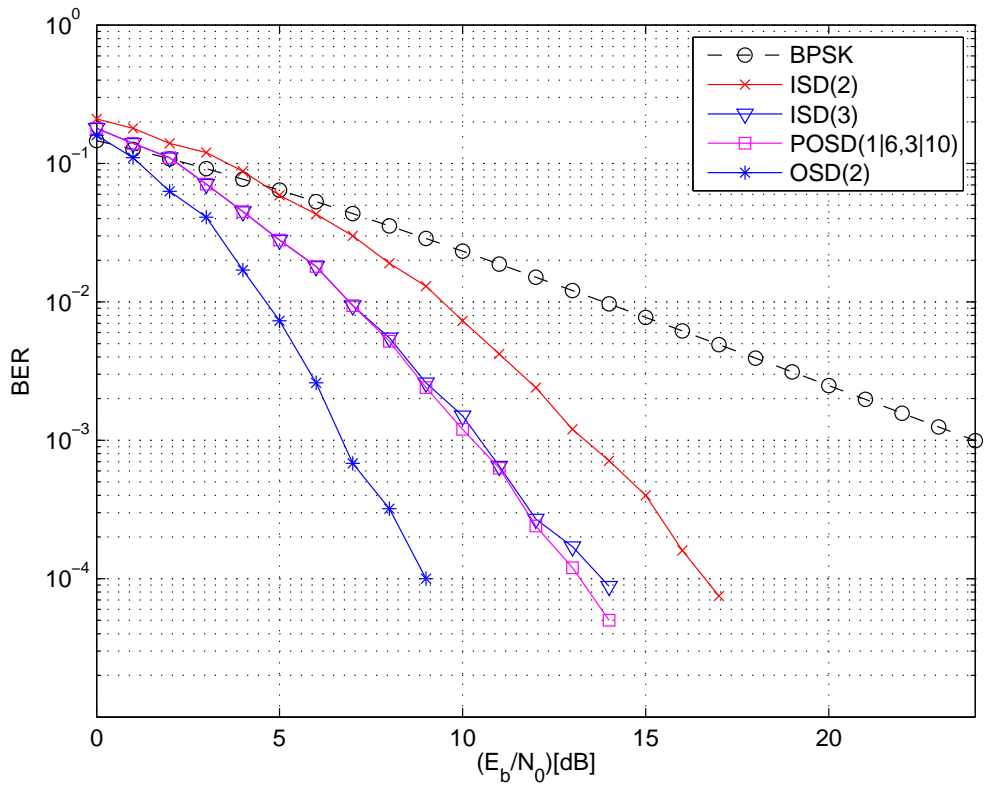}
    \caption{The BER of the $(31,16,7)$ BCH code over a Rayleigh fading
      channel.}
    \label{Fig:3}
  \end{center}
\end{figure}

\begin{figure}[p]
  \begin{center}
    \includegraphics[scale=0.7]{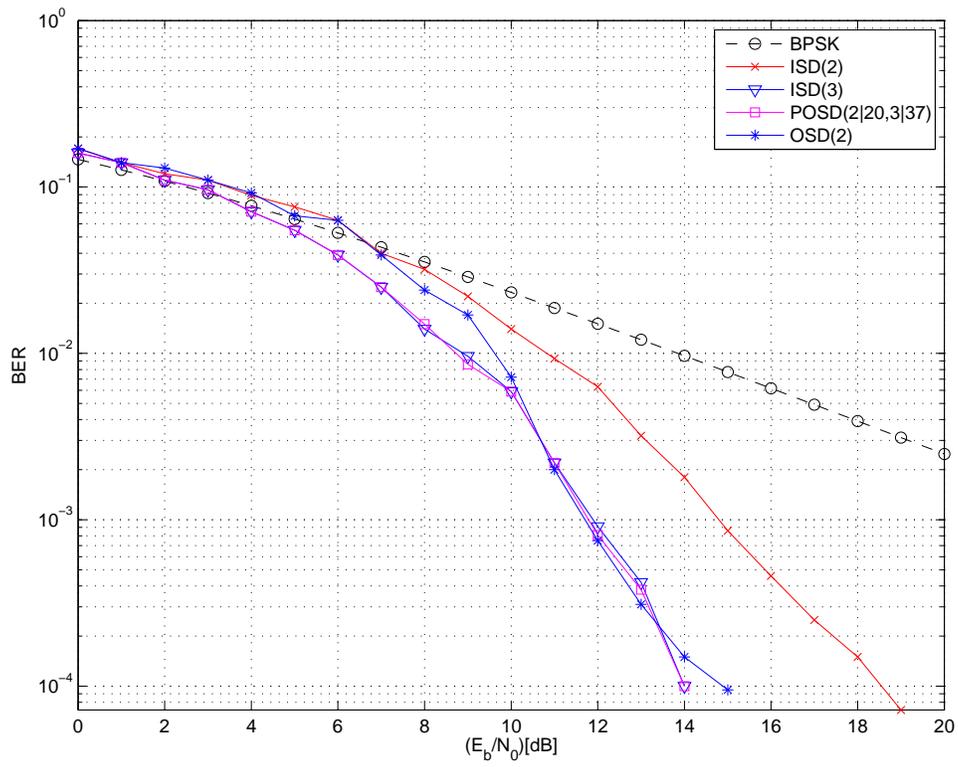}
    \caption{The BER of the $(64,57,4)$ BCH coded over a Rayleigh fading
      channel.}
    \label{Fig:9}
  \end{center}
\end{figure}

\end{document}